\documentclass[11pt]{article}             % Use with LaTeX2e
\usepackage{osid}                         %

\usepackage{amsmath,amssymb,amscd,amsthm,mathrsfs}
\usepackage{bbm}
\usepackage{comment} 
\usepackage{enumerate} 
\usepackage{color}
\usepackage{url}

\theoremstyle{plain}% Theorem-like structures provided by amsthm.sty
\newtheorem{lemma}{Lemma}
\newtheorem{thm}{Theorem}

\newtheorem{corollary}{Corollary}
\newtheorem{proposition}{Proposition}

\theoremstyle{definition}

\newtheorem{remark}{Remark}

\newcounter{app}

\title{Almost All Quantum Channels Are Diagonalizable}
\author{Frederik vom Ende%\thanks{Supported by ...} 
	\\[1mm]{\footnotesize\it Dahlem Center for Complex Quantum
Systems, Freie Universität Berlin, Arnimallee 14, 14195 Berlin, Germany
 \& {frederik.vom.ende@fu-berlin.de}}\\[2ex]
}

\AtBeginDocument{
  \label{CorrectFirstPageLabel}
  
}

\begin{document}

\maketitle
\begin{abstract}
We prove the statement \textit{``The collection of all elements of $\mathcal S$ which have only simple eigenvalues is 
%norm-
dense in $\mathcal S$''} for different sets $\mathcal S$, including:
all quantum channels, the unital channels, the positive trace-preserving maps, all Lindbladians (\textsc{gksl}-generators), and all time-dependent Markovian channels.
Therefore any element from each of these sets can always be approximated by diagonalizable elements of the same set to arbitrary precision.
\end{abstract}

\section{Introduction}
The concept of diagonalization is ubiquitous in quantum physics: it comes up for Hamiltonians (e.g., to find ground-state properties of many-body systems),
%\cite{Laflorencie16} 
quantum channels (e.g., to determine fixed points and find noiseless codes \cite{BNPV10}), the Choi matrix (for obtaining Kraus operators 
%of a channel 
\cite{Watrous18}), and generators of open quantum dynamics (e.g., to find steady states and simplify computation and simulation, etc.~\cite{Albert16}).
What is not as widely known, however, is that quantum channels and Lindbladians can be non-diagonalizable is the sense that their representation matrix cannot always be written as $SDS^{-1}$ for $D$ diagonal and $S$ invertible\footnote{
To prevent potential misunderstandings (as has happened in the past \cite{SL05,SL06,NAJ10,Kosloff21})
%\cite{Kosloff21} (their claim should be wrong, cf. Eq.(A8)),
let us explicitly point out that we are really talking about diagonalizability via similarity transformations and \textit{not} about the---strictly stronger---notion of \textit{unitary} diagonalizability (which is equivalent to normality, i.e.~$\Phi\Phi^\dagger=\Phi^\dagger\Phi$ \cite[Thm.~2.5.4]{HJ1}).
}.
%intuition: Choi always diag because psd (Hermitian)
Lack of diagonalizability can be problematic as it can cause undesirable numerical behavior or complicate potential proof strategies
%indeed, there are a number of papers where things are proven for the special case of diagonalizable maps, cf.~
\cite{Cubitt12,LKM17,CC19,DZP19,CC21,Chrus22}.
%further lit on (non-)diagonalizable GKSL generators:  \cite{DWCM04,Albert14,Kosloff21_2,Samach22} %Samach22: unitarily diag
%, arXiv:2111.12219
A simple example of a non-diagonalizable channel is the qubit channel $\Phi$ defined via
$$
\Phi(\rho):=\frac12\begin{pmatrix}
\rho_{11}+\rho_{22}&-\rho_{11}+\rho_{22}\\
-\rho_{11}+\rho_{22}&\rho_{11}+\rho_{22}
\end{pmatrix}\,,
$$
cf.~\cite{Burgarth13}. Indeed, its Pauli transfer matrix $\mathsf P(\Phi):=(\frac12{\rm tr}(\sigma_j\Phi(\sigma_k)))_{j,k=0}^3$ reads
\begin{equation}\label{eq:Paulitransfer_nondiag}
\mathsf P(\Phi)=\begin{pmatrix}
 1 & 0 & 0 & 0 \\
 0 & 0 & 0 & -1 \\
 0 & 0 & 0 & 0 \\
 0 & 0 & 0 & 0
\end{pmatrix}
\end{equation}
which clearly features a non-trivial Jordan block meaning it cannot be diagonalized.
However, it is also known that in our context not all eigenvalues can admit such Jordan blocks. For example, this is true for the eigenvalue $0$ of a Lindbladian which is always semisimple \cite[App.~A]{VALZ16};
actually, this is true for all peripheral eigenvalues (i.e.~all eigenvalues $\lambda$ with $|\lambda|=1$) of any positive trace-preserving map \cite[Prop.~6.2]{Wolf_Course2012}.
For more examples---also of non-diagonalizable stochastic matrices---refer to, e.g., \cite[Ex.~6]{Chrus22} \& \cite{Lendi87,Wolf_Course2012,PG22}.

A well-known, quite useful property of diagonalizable matrices is, of course, that they are generic, i.e.~every matrix can be approximated by diagonalizable matrices arbitrarily well.
%this property cannot hold for unitary diagonalizability as the set of normal matrices is closed\footnote{
%If $\{A_j\}_{j\in\mathbb N}$ is a sequence of normal matrices with limit $A$, then $(AA^*-A^*A)=\lim_{j\to\infty}(A_jA_j^*-A_j^*A_j)=\lim_{j\to\infty}0=0$ meaning $A$ is normal, as well.
%}
%so the set of normal \textit{channels} is compact as it is the intersection of a closed (normal matrices) and a compact (the quantum channels) set.
%In particular this shows that normal channels cannot approximate anything non-normal; equivalently, all channels sufficiently close to a non-normal (e.g., non-diagonalizable) channel are necessarily non-normal, as well.
However, while this implies that every quantum channel can be approximated by diagonalizable linear maps, it is not clear whether it can be approximated via diagonalizable \textit{channels}.
%naive approach: ``diagonalizable matrices are dense so one can always approximate ...'' (cf., e.g., \cite{Cubitt12}). However, ......
To assure ourselves that not every set of matrices needs to have such an approximation property consider the (obvious) counterexample $\mathcal S:={\rm span}\{|0\rangle\langle 1|\}$ \cite{Hartfiel95}: Trivially, no non-zero element of $\mathcal S$ can be approximated by diagonalizable elements of $\mathcal S$.
Thus it is not immediately clear under what condition a strict subset
%$\mathcal S$ 
of all matrices (such as, e.g., the quantum channels) has the desired approximation property, and while one certainly expects this to be true for channels and related sets it is not obvious how one would prove such a thing.

The purpose of this work is to settle this question via basic tools from perturbation theory (which we review in Sec.~\ref{sec_prelim}).
This will allow us to prove the even stronger statement that almost all channels have only simple eigenvalues\footnote{
Recall that an eigenvalue is called simple if its algebraic multiplicity is $1$.
}: indeed, if all eigenvalues of a matrix are simple, then  they are necessarily distinct which is well known to imply diagonalizability \cite[Thm.~1.3.9]{HJ1}.
With this our main result can be summarized as follows.\medskip

\noindent\textbf{Theorem (Informal).} \textit{Let $\mathcal S$ be one of the following sets: all quantum channels, unital channels, positive trace-preserving maps, all Lindbladians (\textsc{gksl}-generators), or the time-dependent Markovian channels.
Then those elements of $\mathcal S$ which have only simple eigenvalues are norm-dense in $\mathcal S$.}\medskip

In fact, we will prove this result for more general sets $\mathcal S$ but to keep things simple here we refer to Sec.~\ref{sec_mainres} (in particular the paragraphs surrounding the therorems therein) for the precise statements.

\section{Preliminaries}\label{sec_prelim}

First some notation:
given any vector space $\mathcal V$,
the collection of all linear maps $:\mathcal V\to\mathcal V$ will be denoted by $\mathcal L(\mathcal V)$.
If $\mathcal V=\mathbb C^{n\times n}$, then $\mathsf P(n)$ ($\mathsf{PTP}(n)$) denotes the subset of all positive (and trace-preserving) maps;
similarly, $\mathsf{CP}(n)$ ($\mathsf{CPTP}(n)$) is the collection of all completely positive (and trace-preserving) maps.
Recalling the trace norm $\|\cdot\|_1$ which is the sum of all singular values of the input---as well as its ``dual'' norm, i.e.~the operator norm $\|\cdot\|_\infty$ (of matrices) which is given by the largest singular value of the input---one defines 
the operator norm on $\mathcal L(\mathbb C^{n\times n})$ with respect to the trace norm
%will be denoted $\|\cdot\|_{1\to 1}$, i.e.~
via $\|\Phi\|_{1\to 1}:=\sup_{A\in\mathbb C^{n\times n},\|A\|_1=1}\|\Phi(A)\|_1$.
In contrast, the symbol for the diamond norm (completely bounded trace norm) will be $\|\cdot\|_\diamond$, i.e.~$\|\Phi\|_\diamond:=\|\Phi\otimes{\rm id}_n\|_{1\to 1}$ \cite{Watrous18}.
Finally, a well-known yet important fact for our purpose is that in finite dimensions all norms are equivalent \cite[Coro.~5.4.5]{HJ1}; thus whenever we write $\overline{(\cdot)}$---which shall denote the norm closure of a set---we do not have to specify which norm we use.
%To test complete positivity, one usually employs the \textit{Choi matrix}: as proven in the seminal paper of Choi \cite{Choi75} a map $\Phi\in\mathcal L(\mathbb C^{n\times n})$ is completely positive if and only if 
%$\mathsf C(\Phi):=({\rm idZY}\otimes\Phi)(|\Gamma\rangle\langle\Gamma|)$ is positive semi-definite, where $|\Gamma\rangle:=\sum_{j=1}^n|j\rangle\otimes|j\rangle$ is the (unnormalized) maximally entangled state.

Before stating our first lemma we quickly recall the notion of analyticity
\cite[Ch.~8]{Rudin76}: A function $f:I\to\mathbb R$ on an interval $I\subseteq\mathbb R$ is called real analytic if for every $t_0\in I$ there exists $\delta>0$ such that $f(t)=\sum_{n=0}^\infty\frac{f^{(n)}(t_0)}{n!}(t-t_0)^n$ for all $|t-t_0|<\delta$;
obviously, a real analytic function has a unique extension to a holomorphic function on a non-empty connected open set $D\subseteq\mathbb C$.
Moreover, this concept extends to more general co-domains by requiring that every component function of $f$ is real analytic.
Now, having set the stage let us explicitly state the following key result from perturbation theory \cite[Ch.~2, §1]{Kato80} which we will frequently use throughout this paper:
\begin{lemma}\label{lemma_kato_pert}
Let any finite-dimensional vector space $\mathcal V$, any non-empty connected open set $D\subseteq\mathbb C$, and any function $\gamma:D\to\mathcal L(\mathcal V)$ be given.
If $\gamma$ is holomorphic, then the number $s$ of distinct eigenvalues of $\gamma(t)$ is constant if $t\in D$ is not one of the exceptional points (in which case the number of distinct eigenvalues is smaller than $s$).
There are only finitely many exceptional points in each compact subset of $D$.
The number $s$ is equal to $\dim{\mathcal V}$ if and only if there exists $t\in D$ such that $\gamma(t)$ has $\dim{\mathcal V}$ distinct eigenvalues.
All of this continues to hold if $D\subseteq\mathbb R$ is an interval and $\gamma$ is real analytic.
\end{lemma}
\noindent As a side note such exceptional points play an important role in the non-Hermitian Hamiltonian-approach to open quantum systems \cite{MR08,Rotter09,Heiss12,MMCN19}.\medskip
%
%non-diag/singularities/exceptional points important in non-Hermitian quantum systems/PT symmetry: \url{https://doi.org/10.1364/PRJ.396115} \& \url{https://arxiv.org/search/?searchtype=author&query=Starkov%2C+G+A
%}

The first application of this lemma will be a slight generalization of a result first stated by Hartfiel \cite[Coro.~1]{Hartfiel95}.
It will be the key to our main theorems.
While the paper of Hartfiel has been cited in the quantum information literature before \cite{HROWE22,EKC23} it seems that the connection to channel diagonalizability has not been drawn yet.
%recall: eigenvalue is called simple if it has algebraic multiplicity one; thus all ev simple iff all ev distinct
\begin{lemma}\label{lemma_hartfiel}
Given any finite-dimensional vector space $\mathcal V$ and any subset $\mathcal C\subseteq\mathcal L(\mathcal V)$ which is analytically path connected (i.e.~for all $X_1,X_2\in\mathcal C$ there exists a real-analytic function $\gamma:[0,1]\to\mathcal C$ such that $\gamma(0)=X_1$, $\gamma(1)=X_2$) the following statements are equivalent.
\begin{itemize}
\item[(i)] There exists $Z\in\mathcal C$ such that $Z$ has only simple eigenvalues.
\item[(ii)] $
\overline{\{Y\in\mathcal C:Y\text{ has only simple eigenvalues}\}}=\mathcal C
$
\end{itemize}
%Here the closure $\overline{(\cdot)}$ can correspond to any norm on $\mathcal L(\mathcal V)$.
If either of the above statements is true, then $
\overline{\{Y\in\mathcal C:Y\text{ diagonalizable}\}}=\mathcal C$.
 \end{lemma}
\begin{proof}
While the proof idea is the same as in \cite{Hartfiel95} let us sketch it nonetheless:
% for the sake of completeness:
%All we have to show is that every element of $\mathcal C$ can be approximated arbitrarily well by elements of $\mathcal C$ with distinct eigenvalues.
%For this let 
Given any $X\in\mathcal C$,
%be given and 
by assumption there exists \mbox{$\gamma:[0,1]\to\mathcal C$} real analytic such that $\gamma(0)=X$, $\gamma(1)=Z$.
Now Lemma~\ref{lemma_kato_pert} shows that%
%In particular $\gamma$ can be extended to a holomorphic function on some open set in $\mathbb C$ around $[0,1]$
%so by
%Because $\Phi(\varepsilon)$ is obviously holomorphic, b
%\cite[Ch.~2, §1]{Kato80}
%the number of distinct eigenvalues of $\gamma(t)$ is constant almost everywhere;
%a constant aside from the exceptional points (out of which there only finitely many in $[0,1]$ due to compactness)
%indeed, this constant equals $\dim{\mathcal V}$ 
---because $\gamma(1)=Z$ has $\dim{\mathcal V}$ distinct eigenvalues by assumption---all eigenvalues of $\gamma(t)$ are simple for all but finitely many $t\in[0,1]$ (due to $[0,1]$ being compact).
In particular this means that $X=\gamma(0)$ can be approximated by elements of $\mathcal C\supseteq\{\gamma(t):t\in(0,1]\}$ which have only simple eigenvalues.
%The final claim then follows from the well-known fact that all eigenvalues being distinct implies diagonalizability \cite[Thm.~1.3.9]{HJ1}.
\end{proof}
%In particular, there exists $N\in\mathbb N$ such that $\{X(\frac1n)\}_{n\geq N}\subseteq \{Y\in\mathcal C:Y\text{ diagonalizable}\}$.
%This sequence does the job as for any norm $\|\cdot\|$ on $\mathcal L(\mathcal V)$
%one finds $
%\lim_{n\to\infty}\|X-X(\tfrac1n)\|=\lim_{n\to\infty}\tfrac1n\| X-Z \|=0\,.
%$
%
%}\end{proof}
\noindent
The advantage of this more general formulation of the lemma is that it can be applied to convex sets (given $(1-t)X+t Z$ is certainly analytic in $t$) as well as to settings in differential geometry such as, e.g., Lie (semi)groups and general real-analytic manifolds.

We want to emphasize two things here: 1. For the above proof to work it is indeed necessary that $Z$ has only simple eigenvalues. 2. Just because the path $X(t)$ has only finitely many exceptional points that does not mean that the set of channels (which itself is compact) features only finitely or even countably many non-diagonalizable elements.
An example which substantiates both these claims is given by the convex hull of the unital reset channel and the non-diagonalizable channel from the introduction (Eq.~\eqref{eq:Paulitransfer_nondiag}) which is an uncountable set of non-diagonalizable qubit channels, cf.~Sec.~\ref{sec_example} below for more detail.
\section{Main Results}\label{sec_mainres}
%\subsection{Results on channels}
%Now Hartfiel's conclusion from this was that almost all (doubly) stochastic matrices have simple eigenvalues. 
Now on to our main results. 
The advantage 
%Because most sets of linear maps we are interested in are convex, by
Lemma~\ref{lemma_hartfiel} grants us is that it significantly reduces our main problem to the point where we ``only'' have to construct a single element with distinct eigenvalues to conclude that almost all elements from that set  have this property (assuming the set in question is analytically path connected).
One possible construction for the case of unital (i.e.~identity preserving) channels looks as follows:
\begin{lemma}\label{lemma_unital_nsquare}
For all $n\in\mathbb N$ there exists $\Psi\in\mathsf{CPTP}(n)$ such that $\Psi$ is unital and $\Psi$ has $n^2$ distinct eigenvalues.
\end{lemma}
\begin{proof}
Consider the cyclic shift $C:=\sum_{j=1}^n|j\rangle\langle j+1|\in\mathbb C^{n\times n}$ ($|n+1\rangle:=|1\rangle$), the convex combination $A:=\frac12({\bf1}+C)$, and finally the induced linear map $\Psi\in\mathcal L(\mathbb C^{n\times n})$ defined via
$$
\Psi(|j\rangle\langle k|):=\begin{cases}
\sum_{l=1}^nA_{lj}|l\rangle\langle l|&j=k\\
\frac{i}{2^j3^k}|j\rangle\langle k|&j<k\\
-\frac{i}{2^k3^j}|j\rangle\langle k|&j>k
\end{cases}
$$
for all $j,k=1,\ldots,n$.
What remains to show is that $\Psi$ is unital, CPTP, and that is has $n^2$ distinct eigenvalues.

For trace-preservation it---by linearity---suffices to check the basis elements $|j\rangle\langle k|$: indeed, ${\rm tr}(\Psi(|j\rangle\langle k|))=0={\rm tr}(|j\rangle\langle k|)$ for all $j\neq k$ and ${\rm tr}(\Psi(|j\rangle\langle j|))=\sum_{l=1}^nA_{lj}=\frac12+\frac12=1={\rm tr}(|j\rangle\langle j|)$ for all $j$.
Similarly, 
$
\Psi({\bf1})=\sum_{j=1}^n\Psi(|j\rangle\langle j|)=\sum_{l=1}^n(\sum_{j=1}^nA_{lj})|l\rangle\langle l|=\sum_{l=1}^n|l\rangle\langle l|={\bf1}
$
shows unitality 
because, much like before, $\sum_{j=1}^nA_{lj}=\frac12+\frac12=1$ for all $l$.
For complete positivity one readily verifies that the Choi matrix of $\Psi$ reads
$$
%\tau\mathsf C(\Psi)\tau^{-1}=
\underbrace{\begin{pmatrix}
A_{11}&i\delta_{12}&\cdots&i\delta_{1n}\\
-i\delta_{12}&A_{22}&\ddots&\vdots\\
\vdots&\ddots&\ddots&i\delta_{(n-1)n}\\
-i\delta_{1n}&\cdots&-i\delta_{(n-1)n}&A_{nn}
\end{pmatrix}}_{=:X}\oplus\;{\rm diag}(A_{jk})_{j\neq k}
$$
(up to permutation) where $\delta_{jk}:=\frac1{2^j3^k}$.
As $A_{jk}\geq 0$ for all $j,k$ the Choi matrix of $\Psi$ is positive semi-definite (equivalently: $\Psi$ is completely positive \cite{Choi75}) if and only if $X\geq 0$.
Because $X$ is clearly Hermitian it suffices to check that $\|{\bf1}-X\|_\infty\leq1$ (Lemma~\ref{lemma_close_id_psd}, Appendix A):
because $A_{jj}=\frac12$ for all $j$
\begin{align*}
\|{\bf1}-X\|_\infty&\leq\|{\bf1}-{\rm diag}(A_{jj})_j\|_\infty+\|{\rm diag}(A_{jj})_j-X\|_\infty\\
&\leq\frac12+\sum_{j\neq k}\frac1{2^j3^k}\|\,|j\rangle\langle k|\,\|_\infty<\frac12+\sum_{j=1}^\infty\frac1{2^j}\sum_{k=1}^\infty\frac1{3^k}=\frac12+\frac12=1\,,
\end{align*}
as desired.

To see that 
%$\Psi$ is a unital channel we only have to show that 
$\Psi$ has $n^2$ distinct eigenvalues
%. For this 
let us consider the representation matrix $\hat\Psi:=\sum_{a,b,j,k}\langle k|\Psi(|a\rangle\langle b|)|j\rangle\, |j\rangle\langle b|\otimes |k\rangle\langle a|$ of $\Psi$ in the standard basis (obviously, eigenvalues of a linear map coincide with eigenvalues of any matrix representation).
It is easy to see that $\hat\Psi$ up to permutation equals $A\oplus{\rm diag}(i\delta_{jk})_{j>k}\oplus{\rm diag}(-i\delta_{kj})_{j<k}$.
Therefore the eigenvalues of $\Psi$ are given by the eigenvalues of $A$ together with $\{\pm\frac{i}{2^j3^k}:1\leq j<k\leq n\}$.
As the latter set has $2\cdot\frac{n(n-1)}{2}=n^2-n$ distinct elements (due to uniqueness of the prime factorization) it suffices to show that $A$ has $n$ distinct eigenvalues none of which are of the form $\pm\frac{i}{2^j3^k}$.
Indeed, the cyclic shift $C$ has characteristic polynomial $\det(\lambda\cdot{\bf1}-C)=\lambda^n-1$ which yields the well-known result that $C$ has eigenvalues $\{e^{\frac{2\pi ik}n}:k=0,\ldots,n-1\}$
\cite{Hartfiel95} meaning $A$ has eigenvalues $\{\frac12+\frac12e^{\frac{2\pi ik}n}:k=0,\ldots,n-1\}$.
In particular, all eigenvalues $\lambda$ of $A$ are distinct and satisfy $|1-\lambda|\leq1$ because
$$
\Big|1-\Big(\frac12+\frac12e^{\frac{2\pi ik}n}\Big)\Big|=\frac12\big|1-e^{\frac{2\pi ik}n}\big|\leq\frac12\big(1+\big|e^{\frac{2\pi ik}n}\big|\big)=1
$$
for all $k=0,\ldots,n-1$.
In particular, the disk $|1-\lambda|\leq1$ intersects the imaginary axis only in $0\not\in \{\pm\frac{i}{2^j3^k}:1\leq j<k\leq n\}$ which, altogether, concludes the proof.
\end{proof}

%Doing the same for CTO (with H,beta) arbitrary takes care of all the channels with common full rank fixed points
%
%Thus, given any convex subset of channels $\mathcal C\subseteq\mathsf{CPTP}(n)$

As explained before, combining Lemma~\ref{lemma_unital_nsquare} with Lemma~\ref{lemma_hartfiel} immediately yields our first main result because all sets considered therein are convex.
While we only state the following theorem explicitly for (unital) channels and positive trace-preserving maps to keep things simple, let us point out that Theorem~\ref{thm_main1} holds for \textit{any} convex set $S\subseteq\mathcal L(\mathbb C^{n\times n})$ which contains $\Psi$ from Lemma~\ref{lemma_unital_nsquare}.
%(such as, e.g., HPTP, P, CP)
%If we can find a unital channel with only simple eigenvalues this also takes care of P, PTP, CP, CPTP because unital channels are a subset of all of these.

\begin{thm}\label{thm_main1}
Given $\Phi\in\mathsf{CPTP}(n)$, there for every $\varepsilon>0$ exists \mbox{$\Phi_\varepsilon\in\mathsf{CPTP}(n)$} such that $\Phi_\varepsilon$ has only simple eigenvalues
%(is diagonalizable) 
and $\|\Phi-\Phi_\varepsilon\|_\diamond<\varepsilon$.
The same holds when replacing both instances of $\mathsf{CPTP}(n)$ by
%$\mathsf{HP}(n)$, $\mathsf{HPTP}(n)$, $\mathsf{P}(n)$, 
$\mathsf{PTP}(n)$
%, $\mathsf{CP}(n)$, 
or
%$\mathsf{CPTP}_{\bf1 }(n)$.
by the unital channels.
%, and when replacing $\|\cdot\|_\diamond$ by any other norm on $\mathcal L(\mathbb C^{n\times n})$.
\end{thm}
%\begin{proof}
%{...}
%\end{proof}
%\subsection{Results for dynamical generators}

\noindent Some further classes of channels where this proof strategy may work
%---for one reason or another---
are discussed in the outlook (Sec.~\ref{sec_outlook}).

For now let us instead turn to the dynamics side of things:
A natural question at this point is whether generators of open quantum systems can also be approximated by generators with distinct eigenvalues.
And while the set of all such generators is convex---meaning it falls into the domain of Lemma~\ref{lemma_hartfiel}---it is not immediately clear how to construct an element of that set which has only simple eigenvalues.
It turns out, however, that there is a general correspondence between semigroups (e.g., the quantum channels) having the desired approximation property and their infinitesimal generators having the same property.
In order to make this connection precise let us 
%The same holds for the corresponding Lie wedges;
recall that, given some
%real or complex finite-dimensional vector space $\mathcal V$ as well as a 
closed subsemigroup with identity\footnote{
This means that $S$ is a closed subset of $\mathcal L(\mathcal V)$ such that ${\rm id}\in S$
and that for all $\Phi_1,\Phi_2\in S$ it holds that $\Phi_1\Phi_2\in S$.
}
$S\subseteq\mathcal L(\mathcal V)$ one defines its \textit{Lie wedge}
$\mathsf L(S)$ \cite{HHL89} as the collection of all generators of dynamical semigroups in $S$, that is,
\begin{equation*}%\label{eq:liewedge}
\mathsf L(S):=\{A\in\mathcal L(\mathcal V): e^{tA}\in S\text{ for all }t\geq 0\}\,.
\end{equation*}
\begin{proposition}
Let $\mathcal V$ be a finite-dimensional vector space and let $S\subseteq\mathcal L(\mathcal V)$ be a closed, convex semigroup with identity. If
there exists $Y\in S$ which has only simple eigenvalues, then the set of all $L\in\mathsf L(S)$ which have only simple eigenvalues is norm-dense in $\mathsf L(S)$.
\end{proposition}
\begin{proof}
Because $S$ is convex, $X-{\rm id}\in\mathsf L(S)$ for all $X\in S$ (for quantum channels this is also known as ``Markovian approximation'' \cite{Wolf08a}); 
this statement follows from the $C^1$-curve characterization of the Lie wedge (cf.~\cite[Prop.~V.1.7]{HHL89} or \cite[Prop.~2~(iv)]{OSID_thermal_res}): Because $\gamma:[0,1]\to S$, $\gamma(t):=(1-t){\rm id}+tX$---which is well defined as $S$ is convex---is a continuously differentiable curve which starts at the identity its first derivative at zero is in the Lie wedge, i.e.~$\dot\gamma(0)=X-{\rm id}\in\mathsf L(S)$ for all $X\in S$.
Therefore we know that $Y-{\rm id}$ is an element of $\mathsf L(S)$ which has only simple eigenvalues, because the same holds for $Y$ by assumption.
This implies the claim by Lemma~\ref{lemma_hartfiel} together with the well-known fact that $\mathsf L(S)$ is always convex \cite[Prop.~1.14]{LNM1552} as a consequence of the Trotter product formula.
\end{proof}
Thus combining this proposition with Theorem~\ref{thm_main1} immediately yields the corresponding result for dynamical generators; these generators, i.e.~the elements of $\mathsf L(\mathsf{CPTP}(n))$ are often called ``Lindbladian'' or ``\textsc{gksl}-generator''.
\begin{corollary}\label{coro_GKSL_approx}
Given $L\in\mathsf L(\mathsf{CPTP}(n))$, there for every $\varepsilon>0$ exists $L_\varepsilon\in\mathsf L(\mathsf{CPTP}(n))$ such that $L_\varepsilon$ has only simple eigenvalues
%(is diagonalizable).
and $\|L-L_\varepsilon\|_\diamond<\varepsilon$.
The same holds when replacing both instances of $\mathsf{CPTP}(n)$ by
%$\mathsf{HP}(n)$, $\mathsf{HPTP}(n)$, $\mathsf{P}(n)$, 
$\mathsf{PTP}(n)$
%, $\mathsf{CP}(n)$, 
or
%$\mathsf{CPTP}_{\bf1 }(n)$.
by the unital channels.
\end{corollary}
\noindent Much like before, this corollary readily extends from $\mathsf{CPTP}(n)$ to any closed convex subsemigroup $S\subseteq\mathcal L(\mathbb C^{n\times n})$ which contains the identity as well as $\Psi$ from Lemma~\ref{lemma_unital_nsquare}.\medskip

The final set of channels we want to investigate in this work are the so-called Markovian channels. We first turn to the simpler case of \textit{time-independent Markovian} channels, i.e.~those $\Phi\in\mathsf{CPTP}(n)$ which can be written as $\Phi=e^L$ for some $L\in\mathsf L(\mathsf{CPTP}(n))$.
Unlike all sets considered previously the Markovian channels are neither convex nor analytically path-connected so Lemma~\ref{lemma_hartfiel} does not apply here;
this is why we first had to formulate Corollary~\ref{coro_GKSL_approx} so we can ``lift'' the result from the generators to the corresponding channels now:
%This lets us lift this result to Markovian channels:
\begin{lemma}\label{lemma_timeindep_Markov_approx}
Given $\Phi\in\mathsf{CPTP}(n)$ time-independent Markovian, there for all $\varepsilon>0$ exists a time-independent Markovian channel $\Phi_\varepsilon\in\mathsf{CPTP}(n)$ which has only simple eigenvalues such that $\|\Phi-\Phi_\varepsilon\|_\diamond<\varepsilon$.
The same holds when replacing both instances of $\mathsf{CPTP}(n)$ by
%$\mathsf{HP}(n)$, $\mathsf{HPTP}(n)$, $\mathsf{P}(n)$, 
$\mathsf{PTP}(n)$
%, $\mathsf{CP}(n)$, 
or
%$\mathsf{CPTP}_{\bf1 }(n)$.
by the unital channels.
\end{lemma}
\begin{proof}
By assumption there exists $L\in\mathsf L(\mathsf{CPTP}(n))$ such that $\Phi=e^L$, and by Corollary~\ref{coro_GKSL_approx} there exists $L_\varepsilon\in\mathsf L(\mathsf{CPTP}(n))$ such that $\|L-L_\varepsilon\|_\diamond<\frac{\varepsilon}{2}$ and $L_\varepsilon$ has only simple eigenvalues (so $L_\varepsilon\neq 0$).
Then one also finds $\delta\in(0,1]$ such that $e^{tL_\varepsilon}$ for all $t\in(1-\delta,1+\delta)\setminus\{1\}$ has only simple eigenvalues:
to see this, first note that for all $t'>0$ with $t'\max_{\lambda\in\sigma(L_\varepsilon)}|{\rm Im}(\lambda)|<\pi$ (where $\sigma(L_\varepsilon)$ denotes the spectrum of $L_\varepsilon$)
injectivity of $\exp$ on the strip $\{x+iy:x\in\mathbb R, y\in[-\pi,\pi)\}$ implies that all eigenvalues of $e^{t'L_\varepsilon}$ are distinct (because all eigenvalues of $L_\varepsilon$ are distinct and $t'$ is small enough).
Thus Lemma~\ref{lemma_kato_pert} shows that the analytic function $t\mapsto e^{tL_\varepsilon}$ has $n^2$ distinct eigenvalues for all but finitely many exceptional points $t\in[0,2]$. This yields $0<\delta\leq 1$ such that $t\in(1-\delta,1+\delta)\setminus\{1\}$ contains no exceptional point.

In particular this shows that $\Phi_\varepsilon:=e^{(1+\min\{\frac\delta2,\frac\varepsilon{2\|L_\varepsilon\|_\diamond}\})L_\varepsilon}$ has only simple eigenvalues so all that is left is to prove the estimate $\|\Phi-\Phi_\varepsilon\|_\diamond<\varepsilon$.
Using Lemma~\ref{lemma_exp_contractive}~(i) (Appendix~A)
%this follows from the following computation:
we compute
\begin{align*}
\|\Phi-\Phi_\varepsilon\|_\diamond&
\leq\|e^L-e^{L_\varepsilon}\|_\diamond+
\|e^{L_\varepsilon}-
e^{(1+\min\{\frac\delta2,\frac\varepsilon{2\|L_\varepsilon\|_\diamond}\})L_\varepsilon}\|_\diamond\\
&\leq\|L-L_\varepsilon\|_\diamond+\Big\|L_\varepsilon-\Big(1+\min\Big\{\frac\delta2,\frac\varepsilon{2\|L_\varepsilon\|_\diamond}\Big\}\Big)L_\varepsilon\Big\|_\diamond\\
&<\frac{\varepsilon}{2}+\frac\varepsilon{2\|L_\varepsilon\|_\diamond}\|L_\varepsilon\|_\diamond=\varepsilon\,.
\end{align*}
For the unital case one argues analogously.
This proof can be adapted to the case where $\mathsf{CPTP}$ is replaced by $\mathsf{PTP}$ by using Lemma~\ref{lemma_exp_contractive}~(ii) instead of Lemma~\ref{lemma_exp_contractive}~(i): this yields $\Phi_\varepsilon=e^{L_\varepsilon}$ with only simple eigenvalues such that $\|\Phi-\Phi_\varepsilon\|_{1\to 1}<\varepsilon$. Also lowering $\varepsilon$ to $\frac{\varepsilon}{n}$ shows $\|\Phi-\Phi_\varepsilon\|_{\diamond}\leq n\|\Phi-\Phi_\varepsilon\|_{1\to 1}<\varepsilon$ where the first inequality can, e.g., be found in\footnote{
While Paulsen proves this for the completely bounded norm (i.e.~in the Heisenberg picture) this readily transfers to the diamond norm by means of the usual duality.
}
\cite[Prop.~8.11 ff.]{Paulsen03}.
%This concludes the proof.
\end{proof}
Now in a second step we can extend this result from the time-independent to the time-dependent Markovian channels. For this recall that---in accordance with \cite[Thm.~16]{Wolf08a}---$\Phi\in\mathsf{CPTP}(n)$ is called \textit{time-dependent Markovian} if for all $\varepsilon>0$ there exist $L_1,\ldots,L_m\in\mathsf L(\mathsf{CPTP}(n))$, $m\in\mathbb N$ such that $\|\Phi-e^{L_1}\cdot\ldots\cdot e^{L_m}\|_\diamond<\varepsilon$.
The set of all time-dependent Markovian channels will be denoted by $\mathsf{MCPTP}(n)$.
As before this set does not admit any analytic structure
but the fact that the ``building blocks'' of $\mathsf{MCPTP}(n)$ have the desired approximation property (Lemma~\ref{lemma_timeindep_Markov_approx}) allows us to carry over the result.

%Let us point out that---unlike our first main result---the following theorem cannot be directly deduced from Lemma~\ref{lemma_hartfiel} because $\mathsf{MCPTP}(n)$ is not (or at least not obviously) analytically path connected
\begin{thm}
Given $\Phi\in\mathsf{MCPTP}(n)$, there for all $\varepsilon>0$ exists $\Phi_\varepsilon\in\mathsf{MCPTP}(n)$ which has only simple eigenvalues such that $\|\Phi-\Phi_\varepsilon\|_\diamond<\varepsilon$.
%The same result holds when replacing $\mathsf{CPTP}$ by $\mathsf{PTP}$.
\end{thm}
\begin{proof}
By definition there exist $L_1,\ldots,L_m\in\mathsf L(\mathsf{CPTP}(n))$, $m\in\mathbb N$ such that $\|\Phi-e^{L_1}\cdot\ldots\cdot e^{L_m}\|_\diamond<\frac{\varepsilon}{2}$; w.l.o.g.~$L_1,\ldots,L_m\neq 0$. 
%W.l.o.g.~we may assume $m>1$ by the previous corollary.
First, by Lemma~\ref{lemma_timeindep_Markov_approx}
there exists $\tilde L_1\in\mathsf L(\mathsf{CPTP}(n))$ such that $e^{\tilde L_1}$ has only simple eigenvalues and $\|e^{L_1}-e^{\tilde L_1}\|_\diamond<\frac{\varepsilon}{2m}$.
Now if $m=1$ then we can combine the two estimates and we are done. If $m>1$ then we execute the following inductive procedure, starting from $k=1$:

Define the map $\gamma_k:\mathbb C\to\mathcal L(\mathbb C^{n\times n})$ via $\gamma_k(t):=e^{\tilde L_1}\ldots e^{\tilde L_k}e^{tL_{k+1}}$. This map is obviously analytic and has only simple eigenvalues for $t=0$; hence Lemma~\ref{lemma_kato_pert} shows that $\gamma_k(t)$ has only simple eigenvalues for all but finitely many $t\in[0,1]$.
%a constant aside from the exceptional points (out of which there only finitely many in $[0,1]$ due to compactness)
In particular there exists $t_{k+1}\in[1-\frac{\varepsilon}{2m\|L_{k+1}\|_\diamond},1]$ such that $\gamma_k(t_{k+1})$ has only simple eigenvalues. Define $\tilde L_{k+1}:=t_{k+1}L_{k+1}$. If $k<m-1$, then repeat this procedure for $k\to k+1$.

If $k=m-1$, then we found $\tilde L_1,\ldots,\tilde L_m\in\mathsf L(\mathsf{CPTP}(n))$ such that $\Phi_\varepsilon:=e^{\tilde L_1}\cdot\ldots\cdot e^{\tilde L_m}\in\mathsf{MCPTP}(n)$ has only simple eigenvalues. Moreover, the following computation shows that $\Phi_\varepsilon$ is $\varepsilon$-close to $\Phi$:
\begin{align*}
\|\Phi-\Phi_\varepsilon\|_\diamond&<\frac{\varepsilon}{2}+\Big\| \prod_{j=1}^me^{L_j}-\prod_{j=1}^me^{\tilde L_j} \Big\|_\diamond\\
&=\frac{\varepsilon}{2}+\Big\|\sum_{j=1}^m\Big(\prod_{l=1}^{j-1}e^{\tilde L_l}\Big)(e^{L_j}-e^{\tilde L_j})\Big(\prod_{l=j+1}^me^{L_l}\Big)\Big\|_\diamond\\
&\leq\frac{\varepsilon}{2}+\sum_{j=1}^m\Big(\prod_{l=1}^{j-1}\|e^{\tilde L_l}\|_\diamond\Big)\|e^{L_j}-e^{\tilde L_j}\|_\diamond\Big(\prod_{l=j+1}^m\|e^{L_l}\|_\diamond\Big)
\end{align*}
Because all channels have diamond norm $1$ \cite[Prop.~3.44]{Watrous18}, together with Lemma~\ref{lemma_exp_contractive} (Appendix~A) this estimate simplifies to
\begin{align*}
\|\Phi-\Phi_\varepsilon\|_\diamond&<
%\frac{\varepsilon}{2}+\sum_{j=1}^m\|e^{L_j}-e^{\tilde L_j}\|_\diamond\\
%&=
\frac{\varepsilon}{2}+\|e^{L_1}-e^{\tilde L_1}\|_\diamond+\sum_{j=2}^m\|e^{L_j}-e^{\tilde L_j}\|_\diamond\\
&<\frac{\varepsilon}{2}+\frac{\varepsilon}{2m}+\sum_{j=2}^m\|{L_j}-{\tilde L_j}\|_\diamond\\
&=\frac{\varepsilon(m+1)}{2m}+\sum_{j=2}^m (1-t_j)\|L_j\|_\diamond\\
&\leq \frac{\varepsilon(m+1)}{2m}+\sum_{j=2}^m \frac{\varepsilon}{2m\|L_{j}\|_\diamond}\|L_j\|_\diamond=\varepsilon\,.\tag*{\qedhere}
\end{align*}
%For the additional statement one pursues the same strategy together with the well-known fact that $\|\Psi\|_{1\to 1}=1$ for all $\Psi\in\mathsf{PTP}(n)$ \cite[Lemma~1]{Koss72b}.
%Next approximate the constant path $L_c$ def'd via $L_2,L_3,\ldots$ by something analytic ($\varepsilon/3$) $\to$ solution with initial condition $\Phi_\varepsilon$ is analytic (?). But solution started in simple ev so all but finitely many channels of the solution have simple ev + final point is close to $\Phi$ by
%$$
%\|\Phi-\tilde\Phi(T)\|\leq\|\Phi- \Phi_\varepsilon e^{L_2}\cdot\ldots e^{L_m}\|+\|\Phi_\varepsilon e^{L_2}\cdot\ldots e^{L_m}-\tilde\Phi(T)\|\leq\|L_1-\tilde L_1\|+\|L_c-\tilde L\|_1<\frac{2\varepsilon}{3}
%$$
%If final point does not have simple ev, there is a point on the solution with simple ev in the $\varepsilon/3$ ball
%
%$L$ analytic $\Rightarrow$ solution analytic \cite[Ch.~2, Coro.~2.1.1]{DF84}
%
%polynomials are dense in $(L^1([0,t_f]),\|\cdot\|_1)$ (if co-domain is finite dimensional): density of the continuous functions in $L^1$ \cite[Thm.~3.14]{Rudin86} together with the fact that polynomials are dense in the continuous functions by Stone-Weierstrass
\end{proof}
\noindent Finally, as before this theorem applies to the case where $\mathsf{CPTP}$ is replaced by any closed convex semigroup $S\subseteq\mathcal L(\mathbb C^{n\times n})$ which contains the identity as well as $\Psi$ from Lemma~\ref{lemma_unital_nsquare} (such as, e.g., $\mathsf{PTP}$ or the unital channels).
\section{Worked Example}\label{sec_example}

To see that there indeed exist convex subsets of channels which are all non-diagonalizable \& to illustrate this paper's core perturbation argument consider 
the family of channels $\{\Phi(\mu)\}_{\mu\in[0,1]}$ defined by its Pauli transfer matrix
\begin{equation}\label{eq:nondiag_convex_set}
\mathsf P(\Phi(\mu)):=\begin{pmatrix}
1&0&0&0\\0&0&0&-\mu\\0&0&0&0\\0&0&0&0
\end{pmatrix}\,,
\end{equation}
that is, $\Phi(\mu)$ is the convex combination of the unital reset channel $\rho\mapsto{\rm tr}(\rho)\frac{\bf1}2$ and the non-diagonalizable channel from Eq.~\eqref{eq:Paulitransfer_nondiag}.
% which is, obviously, not diagonalizable for any $\lambda\in(0,1]$.
%We will come back to this example in Section~\ref{sec_example} below 

\begin{remark}
Modifying the above strategy lets one easily construct examples of non-diagonalizable unital channels: Given any matrix $\Lambda\in\mathbb R^{3\times 3}$ which is not diagonalizable (over $\mathbb C$), for all $\varepsilon>0$ small enough
$$
\begin{pmatrix}
1&0\\0&\varepsilon\Lambda
\end{pmatrix}
$$
is the Pauli transfer matrix of a non-diagonalizable qubit channel.
This is readily verified via the Fujiwara-Algoet conditions \cite{FA99,RSW02} $|s_1\pm s_2|\leq|1\pm s_3|$ which encode complete positivity in the (real) diagonal entries $s_1\geq s_2\geq |s_3|$ of the Lorentz normal form of $\Lambda$: for unital channels these conditions are necessary and sufficient, and it is obvious that these conditions are satisfied trivially whenever $|s_j|\leq\frac13$ for all $j$.
As an example this procedure yields a \textit{two-dimensional} convex set of (Pauli transfer matrices of) channels which are non-diagonalizable whenever $(a,b)\neq(0,0)$:
$$\Big\{\begin{pmatrix}
1&0&0&0\\0&0&a&0\\0&0&0&b\\0&0&0&0
\end{pmatrix}:a,b\in\mathbb R\text{ such that }|a-b|,|a+b|\leq 1\Big\}$$
Finally note that all examples of non-diagonalizable channels pertain to higher dimensions by extending the channel by, e.g., the identity on the additional diagonal block.
\end{remark}

Let us come back to Eq.~\eqref{eq:nondiag_convex_set} and explicitly approximate these 
%unital 
channels by unital CPTP maps which have only simple eigenvalues.
The perturbing channel from Lemma~\ref{lemma_unital_nsquare} in the case of qubits acts like
$$
\Psi(\rho):=\begin{pmatrix}
\frac12(\rho_{11}+\rho_{22})& \frac{i}{18}\rho_{12} \\
-\frac{i}{18}\rho_{12}  &\frac12(\rho_{11}+\rho_{22})
\end{pmatrix}\,.
$$
Indeed, the convex combination
\begin{equation}\label{eq:convexcomb_pauli_pert}
\mathsf P\big((1-\lambda)\Phi(\mu)+\lambda\Psi  \big)
%=(1-\lambda)\mathsf P(\Phi(\mu))+\lambda\mathsf P(\Psi)
=\begin{pmatrix}
 1 & 0 & 0 & 0 \\
 0 & 0 & \frac{\lambda}{18} & -((1-\lambda) \mu) \\
 0 & -\frac{\lambda}{18} & 0 & 0 \\
 0 & 0 & 0 & 0
\end{pmatrix}
\end{equation}
has eigenvalues $\{1,0,\frac{i}{18}\lambda,-\frac{i}{18}\lambda\}$ which are simple for all $(\lambda,\mu)\in(0,1]\times[0,1]$, which is in agreement with Lemma~\ref{lemma_kato_pert}.
A quick look at the eigenvectors also shows how the perturbation ``fixes'' the problem of the eigenvalue $0$ of $\mathsf P(\Phi(\mu))$ having algebraic multiplicity $3$ but geometric multiplicity $2$:
For $\lambda=0$ the eigenvector $(0,0,(1-\lambda)\mu,\frac{\lambda}{18})^\top$ (of Eq.~\eqref{eq:convexcomb_pauli_pert} to the eigenvalue $0$) can be written as a linear combination of $(0,-i,1,0)^\top$, $(0,i,1,0)^\top$ (corresponding to the eigenvalues $\frac{i}{18}\lambda$, $-\frac{i}{18}\lambda$) which of course becomes impossible as soon as $\lambda>0$.
\section{Outlook}\label{sec_outlook}

In this paper we reviewed the known fact that channels and generators of dynamics can fail to be diagonalizable,
and we demonstrated how tools from perturbation theory can be used
%in quantum information 
to approximate such non-diagonalizable elements by diagonalizable maps ``from the same class'', i.e.~under additional constraints.
We expect that our findings will be a strong, yet simple tool for proving results for quantum channels or for Lindbladians that are stable under small perturbations such as, e.g., approximation results or (non-strict) inequalities.

While we explicitly dealt with unital, Markovian, and general channels
one may instead want to consider channel classes from resource-theoretic approaches such as, e.g., LOCC or (closed) thermal operations. Both these sets are
%well known to be 
convex---cf.~\cite{CLMOW14} and \cite{Gour22,vomEnde22thermal}, respectively---so Lemma~\ref{lemma_hartfiel} applies meaning all one has to do is to construct one element from each set which has only simple eigenvalues.
While this sounds easy at first, the construction for the unital case from Lemma~\ref{lemma_unital_nsquare} does not carry over without further ado;
for the case of thermal operations, for example, one would have to first show that almost all Gibbs-stochastic matrices have only simple eigenvalues---regardless of the exact Gibbs state\footnote{
For those interested: The last point is what makes this non-trivial. While the $\beta$-permutations \cite{Alhambra19,Mazurek19,PolytopeDegen22} are candidates for Gibbs-stochastic matrices with distinct eigenvalues many of them are only valid in certain parameter regimes, i.e.~only if the eigenvalues of the Gibbs state are interrelated in a specific way.
We leave a general construction that is valid for all Gibbs states as an open problem.
}.
%from unital to channels with any fixed point: not clear because no obvious construction for $d$-stochastic matrix which has only simple eigenvalues (which works for all $d>0$) $\to$ cf.~\textit{(Heft S.~26 ff.)}
%
Another problem one could tackle next is to find perturbations of (non-diagonalizable) channels or Lindbladians which, additionally, preserve other spectral properties such as, e.g., the principal eigenspaces;
so far this question has only been looked at for stochastic matrices \cite{PG20,PG22}.
%and could be interesting to extend to quantum channels and the like.

Finally, one may ask whether our ``static'' results can be generalized to true dynamical settings:
One question would be whether (Markovian) dynamics $\{\Phi(t)\}_{t\in[0,T]}\subseteq\mathsf{CPTP}(n)$ can always be approximated by (Markovian) dynamics which have only simple eigenvalues (for all but finitely many times).
%properties (diagonalizable, simple eigenvalues a.e., bijective)? A 
The connection to this paper's techniques is that Lipschitz-continuous dynamics are known to be approximable by analytic (Stinespring) dynamics \cite{vE24_Lipschitz_approx}, where the latter lie in the realm of Lemma~\ref{lemma_kato_pert}. However, 
%---in the case of ... it is not clear how to circumvent the exceptional points
it is not clear how one can, in general, guarantee that these approximating dynamics feature one point
in time where the all eigenvalues are simple---which would be sufficient to guarantee that this holds for almost all times. 
%idea: final points simple ev (how to guarantee?) + find analytic connecting path
%Either way, all these open points show that this work merely scratches the surface and that there are still many interesting questions to explore.

\section*{Acknowledgments}
I would like to thank
%Philippe Faist whose question sparked this paper, to Gunther Dirr for
%making me aware of the paper of Dole\v{z}al \cite{Dolezal64}  (cf.~Remark~\ref{rem_cont_ac_app}, Appendix~\appref{B}), as well as to
%%I am grateful to
%%Gunther Dirr, 
Emanuel Malvetti 
%and Fereshte Shahbeigi
%%, Thomas Schulte-Herbr\"uggen, and Amit Devra
%Sumeet Khatri
%%and the anonymous referee
%for his insight that Lemma~\ref{lemma0}---and thus Theorem~\ref{thm1}---holds even beyond Hermitian matrices.
% i and 
for fruitful discussions during the preparation of this manuscript.
%Moreover I would like to thank Jonas Kitzinger and, again, Sumeet Khatri for proofreading a preliminary version of this manuscript.
This work has been supported by the Einstein Foundation (Einstein Research Unit on Quantum Devices) and the MATH+ Cluster of Excellence.

\section*{Appendix A: Auxiliary lemmata}
This appendix exists so we could outsource two somewhat technical lemmata.
The first one is a well-known sufficient criterion for positive semi-definiteness which we state and prove for the reader's convenience.

\begin{lemma}\label{lemma_close_id_psd}
Let $m\in\mathbb N$ and $X\in\mathbb C^{m\times m}$ Hermitian be given. If $\|{\bf1}-X\|_\infty\leq 1$, then $X$ is positive semi-definite.
\end{lemma}
\begin{proof}
Given any $z\in\mathbb C^m$ we first rewrite $\langle z|X|z\rangle=\langle z|{\bf1}|z\rangle-\langle z|{\bf1}-X|z\rangle$. Because $X$ is Hermitian, $\mathbb R\ni\langle z|{\bf1}-X|z\rangle\leq|\langle z|{\bf1}-X|z\rangle|$ which lets us compute
\begin{align*}
\langle z|X|z\rangle&\geq \|z\|^2-|\langle z|{\bf1}-X|z\rangle|\\
&\geq \|z\|^2-\|z\|\|{\bf1}-X\|_\infty\|z\|=\|z\|^2(1-\|{\bf1}-X\|_\infty)\geq 0\,.\tag*{\qedhere}
\end{align*}
\end{proof}

The second result is concerned with how the distance of infinitesimal generators relates to the distance of their exponentials.
\begin{lemma}\label{lemma_exp_contractive}
For all $n\in\mathbb N$ the following statements hold.
\begin{itemize}
\item[(i)] For all
%$L^1$ functions $L_1,L_2:[0,T]\to \mathsf L(\mathsf{CPTP}(n))$ be given and consider $\dot\Phi_j(t)=L_j(t)\Phi(t)$ for arbitrary but fixed $\Phi(0)\in\mathsf{CPTP}(n)$. Then
$L_1,L_2\in\mathsf L(\mathsf{CPTP}(n))$ one has 
%$$\|\Phi_1(T)-\Phi_2(T)\|_\diamond\leq\int_0^T\|L_1(t)-L_2(t)\|_\diamond\,dt\,.$$
%In particular, 
$\|e^{L_1}-e^{L_2}\|_\diamond\leq\|L_1-L_2\|_\diamond$.
\item[(ii)] For all
%Let $L^1$ functions $L_1,L_2:[0,T]\to \mathsf L(\mathsf{PTP}(n))$ be given and consider $\dot\Phi_j(t)=L_j(t)\Phi(t)$ for arbitrary but fixed $\Phi(0)\in\mathsf{PTP}(n)$. Then
$L_1,L_2\in\mathsf L(\mathsf{PTP}(n))$ one has 
%$$\|\Phi_1(T)-\Phi_2(T)\|_{1\to 1}\leq\int_0^T\|L_1(t)-L_2(t)\|_{1\to 1}\,dt\,.$$
%In particular, 
$\|e^{L_1}-e^{L_2}\|_{1\to 1}\leq\|L_1-L_2\|_{1\to 1}$.
\end{itemize}
\end{lemma}
\begin{proof}
The key to this result is the following identity, which is a special case of Duhamel's formula \cite[Ch.~1, Thm.~5.1]{DF84}:
$
e^A-e^B=\int_0^1 e^{(1-s)B}(A-B)e^{sA}\,ds
$
for all $A,B\in\mathbb C^{m\times m}$.
Thus for any submultiplicative norm on $\mathbb C^{n\times n}$ one finds
\begin{equation}\label{eq:est_int_norm}
\|e^A-e^B\|\leq\int_0^1 \|e^{(1-s)B}\|\,\|A-B\|\,\|e^{sA}\|\,ds\,.
\end{equation}
(i):
%Evaluating Eq.~\eqref{eq:est_int_norm} for the diamond norm and $A=L_1$, $B=L_2$ shows $\|e^{L_1}-e^{L_2}\|_\diamond\leq\int_0^1 \|e^{(1-s)L_2}\|_\diamond\|L_1-L_2\|_\diamond\|e^{sL_1}\|_\diamond\,ds$.
Because $e^{(1-s)L_2},e^{sL_1}\in\mathsf{CPTP}(n)$ for all $s\in[0,1]$ (by assumption on $L_1,L_2$) their diamond norm always equals $1$
% for all $s\in[0,1]$ 
\cite[Prop.~3.44]{Watrous18}.
Thus Eq.~\eqref{eq:est_int_norm} yields $\|e^{L_1}-e^{L_2}\|_\diamond\leq\int_0^1\|L_1-L_2\|_\diamond\,ds=\|L_1-L_2\|_\diamond$, as desired.
For (ii) one pursues the same strategy together with the well-known fact that $\|\Psi\|_{1\to 1}=1$ for all $\Psi\in\mathsf{PTP}(n)$ \cite[Lemma~1]{Koss72b}.
\end{proof}

\bibliographystyle{mystyle}
\bibliography{../../../../control21vJan20.bib}
\end{document}